\documentclass[11pt,a4paper,reqno]{amsart}%

\date{August 27, 2018}
\usepackage{amsthm,amsmath,amsfonts,amssymb,amsxtra,appendix,bookmark,dsfont,bm,mathrsfs,color}
\usepackage{bbm} 

\usepackage[margin=30mm]{geometry}

\newtheorem{theorem}{Theorem}%[section]
\newtheorem{proposition}[theorem]{Proposition}
\newtheorem{lemma}[theorem]{Lemma}
\newtheorem{corollary}[theorem]{Corollary}

\theoremstyle{definition}

\theoremstyle{remark}

\newtheorem{remark}[theorem]{Remark}

\DeclareMathOperator{\Tr}{Tr}
\DeclareMathOperator{\tr}{tr}

\def\bq{\begin{eqnarray}}
\def\eq{\end{eqnarray}}
\def\bqq{\begin{align*}}
\def\eqq{\end{align*}}

\def\nn{\nonumber}

\def\eps{\varepsilon}

\renewcommand{\Re}{\operatorname{Re}}

\newcommand\1{{\ensuremath {\mathds 1} }}

\def\R {\mathbb{R}}

\def\R {\mathbb{R}}

\def\N {\mathbb{N}}

\def\d{{\rm d}}

\title[Improved Lieb--Thirring inequality] {The Lieb--Thirring inequality revisited}

\author[R. L. Frank]{Rupert L. Frank}
\address[R. L. Frank]{Department of Mathematics, LMU Munich, Theresienstrasse 39, 80333 Munich, Germany, and Mathematics 253-37, Caltech, Pasa\-de\-na, CA 91125, USA} 
\email{rlfrank@caltech.edu}

\author[D. Hundertmark]{Dirk Hundertmark}
\address[D. Hundertmark]{Department of Mathematics, Institute for Analysis, Karlsruhe Institute of Technology, 76128 Karlsruhe, Germany, and Department of Mathematics, Altgeld Hall, University of Illinois at Urbana-Champaign, 1409 W. Green Street, Urbana, IL 61801} 
\email{dirk.hundertmark@kit.edu}

\author[M. Jex]{Michal Jex} 
\address[M. Jex]{Department of mathematics, Institute for Analysis, Karlsruhe Institute of Technology, 76128 Karlsruhe, Germany. On leave from Department of Physics, Faculty of \mbox{Nuclear} Sciences and Physical Engineering, Czech Technical University in Prague, B\v{r}ehov\'{a} 7, 11519 Prague, Czech Republic}
\email{michal.jex@kit.edu}

\author[P. T. Nam]{Phan Th\`anh Nam}
\address[P. T. Nam]{Department of Mathematics, LMU Munich, Theresienstrasse 39, 80333 Munich, Germany} 
\email{nam@math.lmu.de}

\begin{document}

\begin{abstract} We provide new estimates on the best constant of the Lieb--Thirring inequality 
for the sum of the negative eigenvalues of Schr\"odinger operators, which significantly improve the so far existing bounds. 
\end{abstract}

\maketitle

\setcounter{tocdepth}{1}
%\tableofcontents
%\addcontentsline{toc}{section}{Contents}

\section{Introduction}

In 1975, Lieb and Thirring \cite{LieThi-75,LieThi-76} proved that  the sum of all negative eigenvalues of Schr\"odinger operators $-\Delta+V$ in $L^2(\R^d)$, with a real-valued potential $V:\R^d\to \R$, admits the bound
\bq \label{eq:LT-sum-eigenvalue}
\Tr[-\Delta + V]_ - \le  L_{1,d} \int_{\R^d} V(x)_-^{1+\frac{d}{2}} \d x
\eq
for a finite constant $L_{1,d}>0$ depending only on the dimension, for all $d\ge 1$. Here we use the convention that $t_\pm=\max\{\pm t,0\}$.

Inequality \eqref{eq:LT-sum-eigenvalue} should be compared with Weyl's law \cite[Theorem 12.12]{LieLos-01} 
\begin{equation}\label{eq:semi-app}
\Tr[-h^2\Delta+V]_-  \approx \frac{1}{(2\pi)^d}\iint_{\R^d\times \R^d} \big[| hk|^2+V(x) \big]_- \d k \d x = L_{1,d}^{\rm cl} h^{-d} \int_{\R^d} V(x)_-^{1+d/2} \d x
\end{equation}
where
$$
L_{1,d}^{\rm cl} = \frac{2}{d+2} \cdot \frac{|B_1|}{(2\pi)^d}
$$
with $|B_1|$ the volume of the unit ball in $\R^d$. While \eqref{eq:semi-app} is only correct in the semiclassical limit $h\to 0$, the Lieb--Thirring inequality \eqref{eq:LT-sum-eigenvalue} is a universal bound for all finite parameters.

A simpler version of \eqref{eq:LT-sum-eigenvalue} is the following bound for a single eigenvalue,
\bq \label{eq:So}
\int_{\R^d}\Big( |\nabla u(x)|^2 + V(x)|u(x)|^2 \Big)\d x \ge - L_{1,d}^{\rm So} \int_{\R^d} V(x)_-^{1+d/2} \d x,
\eq
which is a consequence of Sobolev's inequality, namely some sort of the {uncertainty principle}. This inequality is essentially due to Keller \cite{Keller-61}; see also \cite{CarFraLie-14} for a stability analysis. The Lieb--Thirring inequality \eqref{eq:LT-sum-eigenvalue} extends Sobolev's inequality \eqref{eq:So} by taking the {exclusion principle} into account. 

The {Lieb--Thirring conjecture} \cite{LieThi-76} concerns the best constant in \eqref{eq:LT-sum-eigenvalue} and states that this is given by 
\bq \label{eq:LT-conj}
L_{1,d}=\max\{L_{1,d}^{\rm cl}, L_{1,d}^{\rm So}\}=
\begin{cases}
L_{1,d}^{\rm cl} & \text{ if } d\ge 3, \\
L_{1,2}^{\rm So} & \text{ if } d=1,2,  
\end{cases}
\eq
with $L_{1,d}^{\rm So}$ being the best constant in \eqref{eq:So}. While the lower bound $L_{1,d}\ge\max\{L_{1,d}^{\rm cl}, L_{1,d}^{\rm So}\}$ is obvious, proving the matching upper bound is a major challenge in mathematical physics. 

The original proof of Lieb and Thirring \cite{LieThi-75} gave  $L_{1,d}/L_{1,d}^{\rm cl}\le 4\pi$ in $d=3$. Since then, there have been many contributions devoted to improving the upper bound on $L_{1,d}$ \cite{Lieb-84,EdeFoi-91,BlaStu-96,HunLapWei-00,DolLapLos-08}. 
%The first uniform bound in dimensions was found in 2000 by Hundertmark, Laptev and Weidl \cite{HunLapWei-00} who showed that 
%$$L_{1,d}/L_{1,d}^{\rm cl}\le 2,\quad \forall d\ge 1.$$
The currently best-known result is 
\bq \label{eq:best-known}
L_{1,d}/L_{1,d}^{\rm cl} \le \frac{ \pi}{\sqrt{3}} = 1.814... 
\eq
which was proved for $d=1$ by Eden-Foias in 1991 \cite{EdeFoi-91} and then extended to all $d\ge 1$ by Dolbeault, Laptev and Loss in 2008 \cite{DolLapLos-08}.

Our new result is
\begin{theorem} \label{thm:main} For all $d\ge 1$, the best constant in the Lieb--Thirring inequality \eqref{eq:LT-sum-eigenvalue} satisfies
$$
L_{1,d}/L_{1,d}^{\rm cl} \le  1.456. 
$$
\end{theorem}

Our estimate is a significant improvement over \eqref{eq:best-known}, but in one-dimension is still about $26\%$ bigger than the expected value $L_{1,1}^{\rm So}/L_{1,1}^{\rm cl} = 2/\sqrt{3}=1.155...$ in \cite{LieThi-76}.

Historically, the Lieb--Thirring inequality was invented  to prove the stability of matter \cite{LieThi-75}. In this context, it can be stated as a lower bound on the fermionic kinetic energy,
\bq \label{eq:LT}
\Tr(-\Delta \gamma) \ge K_{d} \int_{\R^d} \gamma(x,x)^{1+\frac{2}{d}}\, \d x.
\eq
Here $\gamma$ is an arbitrary one-body density matrix on $L^2(\R^d)$, i.e. $0\le \gamma\le 1$ with $\Tr \gamma<\infty$, and $\gamma(x,x)$ is the diagonal part of the kernel of $\gamma$ (which can be defined properly by the spectral decomposition). By a standard duality argument, \eqref{eq:LT-sum-eigenvalue}  is equivalent to \eqref{eq:LT}, and the corresponding best constants are related by
\bq \label{eq:L-K}
K_d \left(1+\frac{2}{d}\right) = \left[ L_{1,d} \left(1+\frac{d}{2}\right) \right]^{-2/d}.
\eq
In particular, $K_d$ should be compared with the semiclassical constant 
$$K_d^{\rm cl}=\frac{(2\pi)^{2}}{|B_1|^{2/d}} \cdot \frac{d}{d+2},$$
which emerges naturally from the lowest kinetic energy of the Fermi gas in a finite volume.

In 2011, Rumin \cite{Rumin-11} found a direct proof of \eqref{eq:LT}, without using the dual form \eqref{eq:LT-sum-eigenvalue}. His method has  been used to derive several new estimates, e.g. a positive density analogue of \eqref{eq:LT} in \cite{FraLewLieSei-13}, % and to provide a short proof of the Cwikel-Lieb-Rozenblum (CLR) bound in \cite{Frank-14}.
and it will be also the starting point of our analysis. Note that Rumin's original proof \cite{Rumin-11} gives $K_d/K_d^{\rm cl}\ge d/(d+4)$, and hence 
\begin{align} \label{eq:Rumin-L-version}
L_{1,d}/L_{1,d}^{\rm cl}  \le \left[\frac{d+4}{d}\right]^{d/2},
\end{align}
namely $L_{1,1}/L_{1,1}^{\rm cl}  \le \sqrt{5}=2.236...$ when $d=1$ and and worse
estimates in higher dimensions. Therefore, new ideas are needed to push forward the bound.  

%worse estimates in higher dimensions,  new ideas are needed to push forward the bound. 

Our proof of Theorem \ref{thm:main} contains several main ingredients:

\begin{itemize}

\item First, we will modify Rumin's proof by introducing an {\em optimal momentum decomposition}. This gives $L_{1,1}/L_{1,1}^{\rm cl}  \le 1.618...$ in $d=1$, which is already an improvement over the best-known result \eqref{eq:best-known} in $d=1$.

\item Second, we use the Laptev--Weidl {\em lifting argument} to extend the bound $L_{1,d}/L_{1,d}^{\rm cl}  \le 1.618...$ to arbitrary dimension $d$, which is an improvement over the best-known result \eqref{eq:best-known}. The idea of lifting with respect to dimension is by now classical \cite{
LapWei-00,HunLapWei-00,DolLapLos-08}, but its combination with Rumin's method is not completely obvious and requires an improvement of the bound in \cite{Frank-14}.

\item Third, we take into account a {\em low momentum averaging}. This improves further the bound to $L_{1,1}/L_{1,1}^{\rm cl}  \le 1.456$ in $d=1$ (and worse estimates in higher dimensions). This is one of our key ideas and deviates substantially from Rumin's original argument.

\item Finally, we transfer the one-dimensional bound in the last step to higher dimensions by the {\em lifting argument} again. 
\end{itemize}

These steps will be explained in the next four sections. For the proof of Theorem \ref{thm:main} only the last two sections are relevant, but we feel that a slow presentation of the various new ideas might be useful.

\bigskip

As a by-product of our method we obtain Lieb--Thirring inequalities for fractional Schr\"odinger operators. The inequalities we are interested in have the form
\bq \label{eq:LT-sum-eigenvaluefrac}
\Tr[(-\Delta)^\sigma + V]_ - \le  L_{1,d,s} \int_{\R^d} V(x)_-^{1+\frac{d}{2\sigma}} \d x
\eq
and
\bq \label{eq:LTfrac}
\Tr((-\Delta)^\sigma \gamma) \ge K_{d,\sigma} \int_{\R^d} \gamma(x,x)^{1+\frac{2\sigma}{d}}\, \d x.
\eq
Again, a duality argument shows that the optimal constants in these two inequalities satisfy the relation
\bq \label{eq:L-Kfrac}
K_{d,\sigma} \left(1+\frac{2\sigma}{d}\right) = \left[ L_{1,d,\sigma} \left(1+\frac{d}{2\sigma}\right) \right]^{-\frac{2\sigma}{d}}.
\eq
Finally, the semi-classical constants are given by
\begin{equation}
\begin{split}
K_{d,\sigma}^{\rm cl} 
	& =  \frac{d}{d+2\sigma} \left( \frac{(2\pi)^d}{|B_1|} \right)^{\frac{2\sigma}{d}} \, , \\
L_{1,d,\sigma}^{\rm cl} 
	& = \frac{2\sigma}{d+2\sigma} 
		\frac{|B_1|}{(2\pi)^d}  \, . 
\end{split}
\end{equation}

The main ingredients of the proof of Theorem \ref{thm:main}, except the lifting argument, apply equally well to the fractional case. This gives

\begin{theorem} \label{thm:fractional} For all $d\ge 1$ and $\sigma>0$, the best constant in the Lieb--Thirring inequality \eqref{eq:LTfrac} satisfies
\begin{align*} %\label{eq:kinetic-Rumin-frac}
K_{d,\sigma}/K_{d,\sigma}^{\rm cl}\ge  \max\left\{ \frac{d}{d+4\sigma} 
	\left[ \frac{(d+2\sigma)^2 \sin\left( \frac{2\pi\sigma}{d+2\sigma} \right)}{2\pi\sigma d} \right]^{1+\frac{2\sigma}{d}}, \frac{d}{d+2\sigma}  \left( \frac{2\sigma}{d+2\sigma}\right)^{\frac{4\sigma}{d}} \mathcal{C}_{d,\sigma}^{-\frac{2\sigma}{d}}\right\} 
\end{align*}
where
\begin{align} \label{eq:inf-Cfl-frac}
\mathcal{C}_{d,\sigma} &:= \inf \left\{ \left( \int_0^\infty \varphi^2 \right)^{\frac{d}{2\sigma}} \frac{d}{2\sigma}\int_0^\infty \frac{\Big(1-\int_0^\infty \varphi(s) f(st) \d s\Big)^2}{t^{1+\frac{d}{2\sigma}}}  \d t \right\} 
\end{align}
with the infimum taken over all functions $f,\varphi:\R_+\to \R_+$ satisfying $\int_0^\infty f^2=\int_0^\infty \varphi =1$. 

In particular, when $\sigma=1/2$ and $d=3$, we have 
$\mathcal{C}_{3,1/2} \le 0.046737$
%$\mathcal{C}_{d,\sigma} \le 0.046736...$ 
and hence 
$$K_{3,1/2}/K_{3,1/2}^{\rm cl}\ge 0.826.$$
%$$K_{d,\sigma}/K_{d,\sigma}^{\rm cl}\ge 0.826.$$
\end{theorem}

The proof of Theorem \ref{thm:fractional} is presented in the last section; see also Remark \ref{rem:liftingfrac} in Section \ref{sec:lifting1}.

For $\sigma=1$ and $d>1$, the bound from Theorem \ref{thm:fractional} is not as good as the lower bound in Theorem \ref{thm:main}. For all other cases, Theorem \ref{thm:fractional} yields the best known constants. In particular in the physically relevant case $\sigma=1/2$ and $d=3$, i.e., the ultra--relativistic Schr\"odinger operator in three dimensions, where 
 $K_{3,1/2}^{\rm cl} = \frac{3}{4} (6\pi^2)^{1/3} = 2.923...$, 
%$K_{d,\sigma}^{\rm cl} = \frac{3}{4} (6\pi^2)^{1/3} = 2.923...$, 
our result improves significantly the bounds 
$K_{3,1/2}/K_{3,1/2}^{\rm cl}\ge 0.6$ in \cite[p. 586]{Rumin-11} and 
$K_{3,1/2}K_{d,\sigma}^{\rm cl}\ge 0.558$ in \cite[Eq.(3.4)]{Daubechies-83}. 

An immediate consequence of Theorem \ref{thm:fractional}  is 
\begin{corollary} For every fixed $\sigma>0$, in the limit of large dimensions we have
	\begin{equation} \label{eq:L/Lcl-d-large}
		\limsup_{d\to \infty} L_{1,d,\sigma}/L_{1,d,\sigma}^{\rm cl} \le e.
	\end{equation}	
\end{corollary} 
Indeed, from \eqref{eq:L-Kfrac} we have 
$ L_{1,d,\sigma}/L_{1,d,\sigma}^{\rm cl} = (K_{d,\sigma}^{\rm cl}/K_{d,\sigma})^{\frac{d}{2\sigma}}$.
So \eqref{eq:L/Lcl-d-large} follows from the first lower bound in Theorem \ref{thm:fractional} and the fact that $(\sin(t)/t)^{1/t}\to 1$ as $t\to 0$. Note that Rumin's original proof gives a bound similar to \eqref{eq:L/Lcl-d-large} but with $e$ replaced by $e^2$ (see \eqref{eq:Rumin-L-version}).

As a consequence of \eqref{eq:L/Lcl-d-large} , we also have
\begin{equation} \label{eq:K/Kcl-d-large}
	\lim_{d\to \infty} K_{d,\sigma}/K_{d,\sigma}^{\rm cl} =1.
	\end{equation}
	The lower bound $\liminf_{d\to\infty} K_{d,\sigma}/K_{d,\sigma}^{\rm cl} \ge 1$  
	follows from \eqref{eq:L/Lcl-d-large} and the upper bound 
	$K_{d,\sigma}/K_{d,\sigma}^{\rm cl}\le 1$ is 
	well-known, see \cite{Frank-14}.

%
%Consequently, 
%	\begin{equation} \label{eq:K/Kcl-d-large}
%	\lim_{d\to \infty} K_{d,\sigma}/K_{d,\sigma}^{\rm cl} =1.
%	\end{equation}
%
%\eqref{eq:Rumin-L-version}
%
%\begin{proof} The first our lower bound in 
%	One always has the upper bound $ K_{d,\sigma}/K_{d,\sigma}^{\rm cl}\le 1 $. 
%	This follows from well--known semi--classical bounds for the Lieb--Thirring inequality, for a direct argument see \cite{Frank-14}.  
%	Since $\sin(x)/x\to 1$ as $x\to 0$, the lower bound from Theorem \ref{thm:fractional} shows 
%	\begin{align*}
%		\liminf_{d\to\infty} K_{d,\sigma}/K_{d,\sigma}^{\rm cl} \ge 1.
%	\end{align*}
%\end{proof}

Finally, we note that in 2013, Lundholm and Solovej \cite{LunSol-13} found another direct proof of the kinetic estimate \eqref{eq:LT}. Their approach is based on a local version of the exclusion principle, which is inspired by the first proof of the stability of matter by Dyson and Lenard \cite{DysLen-67}. Recently, the ideas in \cite{LunSol-13} have been developed further in \cite{Nam-18} to show that 
\bq \label{eq:LT-gradient}
\Tr(-\Delta \gamma) \ge (K_{d}^{\rm cl}-\eps) \int_{\R^d} \gamma(x,x)^{1+\frac{2}{d}}\, \d x - C_{d,\eps} \int_{\R^d} |\nabla \sqrt{\gamma(x,x)}|^2 \d x 
\eq
for all $d\ge 1$ and $\eps>0$ (the gradient error term is always smaller than the kinetic term \cite{HO-77}). Note that from \eqref{eq:LT-gradient}, as well as from all existing proofs of the Lieb--Thirring inequality (including the present paper), the real difference between dimensions is not visible. Therefore, new ideas are certainly needed to attack the full conjecture \eqref{eq:LT-conj}.

%We present the proof in Section \ref{sec:momentum-decomposition}; see also Remark \ref{rem:liftingfrac} in Section \ref{sec:lifting1}.
%
%\begin{align} \label{eq:Affrac}
%A_f^{(\sigma)}:=  \frac{d}{2\sigma}\int_0^\infty \frac{(1-f(t))^2}{t^{1+\frac{d}{2\sigma}}}  \d t.
%\end{align}
%
%
%
%On the other hand, in the pseudo-relativistic case $s=1/2$ our bound gives
%$ K/K_{\rm cl} \ge 0.7655$ which is significantly better than the best known result $0.5479$  by Daubechies \cite{Daubechies-83}. 
%
%Relativistic, 3D Rumin  $K\ge  1.754$, which is slightly stronger than the constant $1.63$ given in \cite[Eq.
%(3.4)]{Daubechies-83}. Semicalssical constant
%$$
%K_3^{(1/2)} = \frac{3}{4} (6\pi^2)^{1/3} = 2.923333... 
%$$
%
%$$1.754/2.923333=0.6, 1.63/2.923333=0.5575827318$$

\subsection*{Acknowledgment.} We thank Sabine Boegli for helpful discussions. This work was partially supported by U.S. NSF grant DMS-1363432 (R.L.F.),  the Alfried Krupp von Bohlen und Halbach Foundation, and the Deutsche Forschungs\-gemein\-schaft 
 (DFG) through CRC 1173 (D.H.). 

%%%%%%%%%%%

\section{Optimal momentum decomposition} \label{sec:momentum-decomposition}

In this section, we use a modified version of Rumin's proof in \cite{Rumin-11} to prove

\begin{proposition} \label{prop:gRumin} For $d\ge 1$, the best constant in the Lieb--Thirring inequality \eqref{eq:LT} satisfies
$$
K_d/K_d^{\rm cl}\ge   \frac{d}{d+4} \left[ \frac{(d+2)^2\sin\left(\frac{2\pi}{d+2}\right)}{2\pi d} \right]^{1+\frac{2}{d}}.
%K_d/K_d^{\rm cl}\ge   \frac{d}{d+4} \left[ \Gamma\left( \frac{d+4}{d+2}\right) \Gamma\left( \frac{2d+2}{d+2}\right)  \right]^{-1-\frac{2}{d}}.
$$
%$$
%K_d/K_d^{\rm cl}\ge   \frac{d}{d+4} \left[ \Gamma\left( \frac{d+4}{d+2}\right) \Gamma\left( \frac{2d+2}{d+2}\right)  \right]^{-1-\frac{2}{d}}.
%$$
%or equivalently
%$$
%\frac{L}{L_{\rm cl}} \le  \left[ 1+ \frac{4s}{d} \right] ^{\frac{d}{2s}} \left[ \Gamma\left( 2-\frac{d}{d+2s}\right) \Gamma\left( 2- \frac{2s}{d+2s}\right)  \right]^{1+ \frac{d}{2s}}. 
%$$
%In particular, when $d=1$ we get $K_1^{\rm op}/K_1^{\rm cl}\ge \frac{2187\sqrt{3}{320 \pi^3} =  0.381777...$ and $L_{1,d}^{\rm op}/L_{1,d}^{\rm cl}\le \tfrac{1}{5}(\tfrac{9\sqrt{3}}{4\pi}) \le 1.618435...$. 
%

In particular, when $d=1$ we get $K_1/K_1^{\rm cl}\ge \frac{2187\sqrt{3}}{320 \pi^3} \ge  0.381777$ and $L_{1,1}/L_{1,1}^{\rm cl}\le 1.618435$. 
\end{proposition}

%$\frac{2187\sqrt{3}{320 \pi^3}$

%On the other hand, in the pseudo-relativistic case $s=1/2$ our bound gives
%$ K/K_{\rm cl} \ge 0.7655$ which is significantly better than the best known result $0.5479$  by Daubechies \cite{Daubechies-83}. 

%In the non-relativistic case $s=1$, the best known result is 
%$$\frac{K}{K_{\rm cl}} \ge (\sqrt{3}/\pi)^{2s/d}$$
%which was proved by Eden-Foias \cite{EdeFoi-91} for $d=1$ and by Dolbeault-Laptev-Loss \cite{DolLapLos-08} for all $d\ge 1$. 
%%In the most physically interesting case $d=3$, it is $0.672\times K_d^{\rm cl}$.  %See also  \cite{Levitt-14} for a numerical investigation. 

\begin{proof} Let $\gamma$ be an operator on $L^2(\R^d)$ with $0\leq\gamma\leq 1$. By a density argument, it suffices to consider the case when $\gamma$ is a finite-rank operator with smooth eigenfunctions. For any function $f:\R_+\to \R_+$ with $\int_0^\infty f^2=1$, using the momentum decomposition
$$
-\Delta=p^2= \int_0^\infty f^2(s/p^2)\d s, \quad  p=-i\nabla,
$$ 
and Fubini's theorem we can write 
\begin{align}\label{eq:kinetic-representation}
\Tr(-\Delta \gamma)= \int_0^\infty \Tr [f(s/p^2)\gamma f(s/p^2)] \d s= \int_{\R^d} \left[ \int_0^\infty (f(s/p^2)\gamma f(s/p^2))(x,x) \d s \right] \d x.
\end{align}

Next, we estimate the kernel of $f(s/p^2)\gamma f(s/p^2)$.  Using Cauchy--Schwarz and $0\le \gamma\le 1$, for every $\eps>0$ we have the operator inequalities
\begin{align}\label{eq:CS-operator}
\gamma &\le (1+\eps) f(s/p^2)\gamma f(s/p^2) + (1+\eps^{-1})  (1-f(s/p^2)) \gamma (1-f(s/p^2)) \nn\\
&\le (1+\eps) f(s/p^2)\gamma f(s/p^2) + (1+\eps^{-1})  (1-f(s/p^2))^2 . 
\end{align}
This inequality implies for any $x\in\R^d$ the kernel bound
\begin{equation}
\label{eq:kernelbound}
\gamma(x,x) \leq (1+\eps) (f(s/p^2)\gamma f(s/p^2))(x,x) + (1+\eps^{-1})  (1-f(s/p^2))^2(x,x) .
\end{equation}
Optimizing over $\eps>0$ we obtain
\begin{align} \label{eq:CS-kernel}
\sqrt{\gamma(x,x)} \le \sqrt{(f(s/p^2)\gamma f(s/p^2))(x,x)} + \sqrt{(1-f(s/p^2))^2 (x,x)}. 
\end{align}
Moreover, it is straightforward to see that 
\begin{align} \label{eq:rho-}
(1-f(s/p^2))^2 (x,x) = \frac{1}{(2\pi)^{d}} \int_{\R^d} (1-f(s/k^2))^2 \d k = s^{\frac{d}{2}}  \frac{|B_1|}{(2\pi)^d} A_f
\end{align}
where 
\begin{align} \label{eq:Af}
A_f:=  \frac{d}{2}\int_0^\infty \frac{(1-f(t))^2}{t^{1+\frac{d}{2}}}  \d t.
\end{align}
Consequently, we deduce from \eqref{eq:CS-kernel} that
\begin{align} \label{eq:bound-Af}
(f(s/p^2)\gamma f(s/p^2))(x,x)  \ge  \left[\sqrt{\gamma(x,x)}-\sqrt{s^{\frac{d}{2}}  \frac{|B_1|}{(2\pi)^d} A_f} \, \right]_+^2.
\end{align}
Next,  inserting \eqref{eq:bound-Af} into \eqref{eq:kinetic-representation} and integrating over $s>0$ lead to
\begin{align} \label{eq:Rumin-bound}
\Tr(-\Delta \gamma) \ge \left( \int_{\R^d}  \gamma(x,x)^{1+\frac{2}{d}} \d x \right) \left( \frac{|B_1|}{(2\pi)^d} A_f\right)^{-\frac{2}{d}} \frac{d^2}{(d+2)(d+4)} .
\end{align}
Thus,
\begin{align} \label{eq:Rumin-bound-K}
K_d/K_d^{\rm cl} \ge \frac{d}{d+4}  \left( A_f\right)^{-\frac{2}{d}}. 
\end{align}

Finally, it remains to minimize $A_f$ under the constraint $\int_0^\infty f^2 =1$. We note that the proof in \cite{Rumin-11} corresponds to $f(t)=\1(t\le 1)$ (although the representation there is rather different), which gives $A_f=1$ but this is not optimal. From Lemma \ref{lem:min-Af} below we have 
$$
\inf_f A_f = \left[ \frac{d}{d+2} \frac{\frac{2\pi}{d+2}}{\sin\left(\frac{2\pi}{d+2}\right)}\right]^{1+\frac{d}{2}}. 
$$
Inserting this into \eqref{eq:Rumin-bound-K} we conclude the proof of Proposition \ref{prop:gRumin}.
\end{proof}

In the previous proof we needed the following solution of a minimization problem.

\begin{lemma} \label{lem:min-Af}
For any constant $\beta>1$,
$$
\inf\left\{ \int_0^\infty (1-f(t))^2 \,t^{-\beta}\,dt :\ f:\R_+\to \R_+, \int_0^\infty f^2\,dt = 1 \right\} 
	= \frac{(\beta-1)^{\beta-1}}{\beta^\beta}
	\left(\frac{\pi/\beta}{\sin(\pi/\beta)}
	\right)^\beta
$$
%$$
%\inf\left\{ \int_0^\infty (1-f(t))^2 \,t^{-\beta}\,dt :\ f:\R_+\to \R_+, \int_0^\infty f^2\,dt = 1 \right\} = \frac{\big[\Gamma(1+1/\beta)\ \Gamma(2-1/\beta)\big]^\beta}{\beta -1}
%$$
and equality is achieved if and only if
$$
f(t) = \frac{1}{1+\mu t^\beta}
\qquad\text{with}\qquad 
\mu = \left[\frac{\beta-1}{\beta}\cdot \frac{\pi/\beta}{\sin\left(\pi/\beta\right)}\right]^\beta.
$$
%$$
%f(t) = \frac{1}{1+\mu t^\beta}
%\qquad\text{with}\qquad 
%\mu = \big[\Gamma(1+1/\beta) \Gamma(2-1/\beta)\big]^\beta.
%$$
\end{lemma}

\begin{proof} Heuristically, the optimizer can be found by solving the Euler--Lagrange equation, but to make this rigorous one would have to prove that a minimizer exists. 
 This can be easily done by setting $h(t)= (1-f(t))t^{-\beta/2} $, so the minimization problem is equivalent to 
$$
\inf\left\{\int_0^\infty h(t)^2\, dt :\, h\in \partial C \right\}
$$
where $\partial C = \{h:\R_+\to\R ,\, \int_0^\infty (1-t^{\beta/2}h(t))^2\, dt =1\}$ is the boundary of the strictly convex set $C= \{h:\R_+\to\R ,\, \int_0^\infty (1-t^{\beta/2}h(t))^2\, dt \le 1\}$. Since $C$ is closed, which follows easily from Fatou's lemma, and does not contain the zero function, it contains a functions $h_*$ of minimal length. Necessarily $h_*\in \partial C$, otherwise $h_*$ would be in the interior of $C$ and we could shrink it, thus reducing its length a little bit, which is impossible. 
So $h_*(t)= (1-f_*(t))t^{-\beta/2} $ has minimal $L^2$ norm under all $f$ 
with $\int_0^\infty f(t^2\, dt= \int_0^\infty (1-t^{\beta/2}h(t)^2\, dt=1$. 
Hence $f_*$  is a minimizer which must obey the Euler--Lagrange equation. 

A more direct solution is as follows:  
Let $f_*(t) = (1+(\mu_*t)^\beta)^{-1}$ with
$$
\mu_* =  \int_0^\infty \frac{\d t}{(1+t^\beta)^2} \,,
$$
so that $t^{-\beta} (1-f_*(t))=\mu_*^\beta f_*(t)$ and 
$$
\int_0^\infty f_*(t)^2\,\d t = \int_0^\infty \frac{\d t}{(1+\mu_*t^\beta)^2} =  \mu_*^{-1}\int_0^\infty \frac{\d t}{(1+t^\beta)^2} = 1\,.
$$
We see that for any $f:\R_+\to \R_+$ with $\int_0^\infty f(t)^2\,\d t =1$,
\begin{align*}
& \int_0^\infty t^{-\beta} (1-f(t))^2 \,\d t - \int_0^\infty t^{-\beta} (1-f_*(t))^2 \,\d t \\
& = 2 \int_0^\infty t^{-\beta} (1-f_*(t))(f_*(t)-f(t))\,\d t + \int_0^\infty t^{-\beta} (f(t)-f_*(t))^2 \,\d t \\
& = 2 \mu_*^\beta \int_0^\infty  f_*(t)(f_*(t)-f(t))\,\d t + \int_0^\infty t^{-\beta} (f(t)-f_*(t))^2\,\d t \\
& =  \mu_*^\beta \int_0^\infty  (f_*(t)-f(t))^2 \,\d t + \int_0^\infty t^{-\beta} (f(t)-f_*(t))^2 \,\d t \geq 0 \,.
\end{align*}
Here we used $t^{-\beta} (1-f_*(t))=\mu_*^\beta f_*(t)$ in the third identity and  $\int_0^\infty f_*^2= \int_0^\infty f^2 = \tfrac{1}{2}\int f_*^2 + \tfrac{1}{2}\int_0^\infty f^2$ in the last one. This shows that the infimum is attained if and only if $f=f_*$. 

It remains to compute the infimum and $\mu_*$. 
Both follow from the formula \cite[Abramowitz--Stegun, 6.2.1 and 6.2.2]{AbrSte-64}
$$ 
\int_0^\infty \frac{u^\zeta}{(1+u)^2}\,du = \Gamma(1+\zeta)\, \Gamma(1-\zeta)
\qquad \text{if}\quad
-1<\Re \zeta<1 \,.
$$
Alternatively one can use a keyhole type contour encircling the positive real axis 
and the residue theorem,  see  \cite[Section 11.1.III]{BN}, to directly evaluate 
$
\int_0^\infty \frac{u^\zeta}{(1+u)^2}\,du 
$. 

Letting $u=t^\beta$, we have
$$
 \mu_* = \int_0^\infty \frac{\d t}{(1+t^\beta)^2} = \frac{1}{\beta} \int_0^\infty \frac{u^{1/\beta-1}\,du}{(1+u)^2} =  \frac{\Gamma(1/\beta)\, \Gamma(2-1/\beta)}{\beta} 
$$
The functional equations $\Gamma(1+z)= z\Gamma(z)$ and $\Gamma(z)\Gamma(1-z)= \frac{\pi}{\sin(\pi z)}$, the last one again valid for $-1<\Re z <1$, yield 
$$
\mu_* = \frac{1}{\beta}\Big(1-\frac{1}{\beta}\Big)\Gamma(1/\beta)\Gamma(1-1/\beta) = \Big(1-\frac{1}{\beta}\Big)\frac{\pi/\beta}{\sin\left(\pi/\beta\right)}
$$
Moreover,
\begin{align*}
\int_0^\infty (1-f_*(t))^2 t^{-\beta}\,\d t = \mu_*^\beta  \int_0^\infty \frac{(\mu_*t)^\beta\,\d t}{(1+\mu_* t^\beta)^2} = \mu_*^{\beta-1} \int_0^\infty \frac{t^\beta \,\d ts}{(1+t^\beta)^2}
\end{align*}
and
\begin{align*}
\int_0^\infty \frac{t^\beta \,\d t}{(1+t^\beta)^2} 
	& = \frac{1}{\beta} \int_0^\infty \frac{u^{1/\beta}\,du}{(1+u)^2} = \frac{\Gamma(1+1/\beta)\,\Gamma(1-1/\beta)}{\beta} 
		= \frac{\Gamma(1/\beta)\,\Gamma(1-1/\beta)}{\beta^2} \\
	& = \frac{1}{\beta} \frac{\pi/\beta}{\sin\left( \pi/\beta \right)}.
\end{align*}
This proves the claimed formula.

%It remains to compute the infimum and to express $\mu_*$ in terms of Gamma functions. Both follow from the formula \cite[Abramowitz--Stegun, 6.2.1 and 6.2.2]{AbrSte-64}
%$$
%\int_0^\infty \frac{u^\zeta}{(1+u)^2}\,du = \Gamma(1+\zeta)\, \Gamma(1-\zeta)
%\qquad\text{if}\qquad
%-1<\Re \zeta<1 \,.
%$$
%Indeed, letting $u=t^\beta$, we have
%$$
%\mu_*^{1/\beta} = \int_0^\infty \frac{\d t}{(1+t^\beta)^2} = \frac{1}{\beta} \int_0^\infty \frac{u^{1/\beta-1}\,du}{(1+u)^2} =  \frac{\Gamma(1/\beta)\, \Gamma(2-1/\beta)}{\beta} = \Gamma(1+1/\beta) \,\Gamma(2-1/\beta) \,.
%$$
%Moreover,
%\begin{align*}
%\int_0^\infty (1-f_*(t))^2 t^{-\beta}\,\d t = \mu_*^2 \int_0^\infty \frac{t^\beta\,\d t}{(1+\mu_* t^\beta)^2} = \mu_*^{1-1/\beta} \int_0^\infty \frac{t^\beta \,\d ts}{(1+t^\beta)^2}
%\end{align*}
%and
%$$
%\int_0^\infty \frac{t^\beta \,\d t}{(1+t^\beta)^2} = \frac{1}{\beta} \int_0^\infty \frac{u^{1/\beta}\,du}{(1+u)^2} = \frac{\Gamma(1+1/\beta)\,\Gamma(1-1/\beta)}{\beta} = \frac{\Gamma(1+1/\beta)\,\Gamma(2-1/\beta)}{\beta -1} \,.
%$$
%This proves the claimed formula.
\end{proof}

%%%%%%%%%%%%%

\section{Lifting to higher dimensions. I}\label{sec:lifting1}

In dimension $d=1$ Proposition \ref{prop:gRumin} yields $L_{1,1}/L_{1,1}^{\rm cl}\leq 1.618435$, which is better than for instance the bound in dimension $d=3$, namely $L_{1,3}/L_{1,3}^{\rm cl}\leq 1.994584$. In this section we use a procedure of Laptev and Weidl \cite{Laptev-97,LapWei-00} to show that the higher-dimensional fraction $L_{1,d}/L_{1,d}^{\rm cl}$ is at least as good as the low-dimensional one.

The idea is to consider potentials $V$ on $\R^d$ that take values in the self-adjoint operators on some separable Hilbert space $\mathcal H$. We are looking for an inequality of the form
\bq \label{eq:LT-sum-eigenvalueop}
\Tr[-\Delta + V]_ - \le  L_{1,d}^{\rm op} \int_{\R^d} \tr \left( V(x)_-^{1+\frac{d}{2}}\right) \d x \,,
\eq
where $\tr$ denotes the trace in $\mathcal H$, $\Tr$ the trace in $L^2(\R^d;\mathcal{H})= L^2(\R^d)\otimes\mathcal{H}$, the operator $-\Delta$ is interpreted as $-\Delta\otimes\1_\mathcal{H}$, and where, by definition, the constant $L_{1,d}^{\rm op}$ is independent of $\mathcal H$. Taking $\mathcal H$ one-dimensional we see that \eqref{eq:LT-sum-eigenvalueop} coincides with \eqref{eq:LT-sum-eigenvalue} and therefore 
\bq \label{eq:L<=Lop}
L_{1,d}\leq L_{1,d}^{\rm op}\,.
\eq
It is not known whether $L_{1,d}$ and $L_{1,d}^{\rm op}$ coincide, but in this section we will show that the upper bound on $L_{1,d}$ from Proposition \ref{prop:gRumin} is, in fact, also an upper bound on $L_{1,d}^{\rm op}$.

We show this by using the classical duality argument. This shows the analogue of \eqref{eq:L-K}, that is,
\bq \label{eq:L-Kop}
K_d^{\rm op} \left(1+\frac{2}{d}\right) = \left[ L_{1,d}^{\rm op} \left(1+\frac{d}{2}\right) \right]^{-2/d} \,,
\eq
where $K_d^{\rm op}$ denote the best constant in the inequality
\bq \label{eq:LTop}
\Tr(-\Delta \gamma) \ge K_{d}^{\rm op} \int_{\R^d} \tr \left(\gamma(x,x)^{1+\frac{2}{d}}\right) \d x.
\eq
for all operators $\gamma$ on $L^2(\R^d;\mathcal H)$ satisfying $0\leq\gamma\leq 1$, where $\mathcal H$ is an arbitrary (separable) Hilbert space. For such $\gamma$, one can consider $\gamma(x,x)$ as a non-negative operator in $\mathcal H$.

The following proof improves an argument from \cite{Frank-14}.

\begin{proposition} \label{prop:gRuminlift}
For $d\ge 1$, the best constant in the Lieb--Thirring inequality \eqref{eq:LTop} satisfies
$$
K_d^{\rm op}/K_d^{\rm cl}\ge   \frac{d}{d+4} \left[ \frac{(d+2)^2\sin\left(\frac{2\pi}{d+2}\right)}{2\pi d} \right]^{1+\frac{2}{d}}.
%K_d/K_d^{\rm cl}\ge   \frac{d}{d+4} \left[ \Gamma\left( \frac{d+4}{d+2}\right) \Gamma\left( \frac{2d+2}{d+2}\right)  \right]^{-1-\frac{2}{d}}.
$$
In particular, when $d=1$ we get $K_1^{\rm op}/K_1^{\rm cl}\ge 0.381777$ and $L_{1,d}^{\rm op}/L_{1,d}^{\rm cl}\le 1.618435$.
%In particular, when $d=1$ we get $K_1^{\rm op}/K_1^{\rm cl}\ge \tfrac{1}{5}(\tfrac{9\sqrt{3}}{4\pi})^3 \ge  0.38177...$ and $L_{1,d}^{\rm op}/L_{1,d}^{\rm cl}\le 5^{1/2}(\tfrac{4\pi}{9\sqrt{3}})^{3/2} \le 1.61844...$.
\end{proposition}

\begin{proof}
Let $\gamma$ be an operator on $L^2(\R^d;\mathcal H)$ with $0\leq\gamma\leq 1$. By a density argument we may assume that $\mathcal H$ is finite-dimensional and that $\gamma$ is finite rank and with smooth eigenfunctions. The analogue of \eqref{eq:kinetic-representation} is
\begin{align}\label{eq:kinetic-representationop}
\Tr(-\Delta \gamma)= \int_{\R^d} \tr \left[ \int_0^\infty (f(s/p^2)\gamma f(s/p^2))(x,x) \d s \right] \d x
\end{align}
for any $f:\R_+\to \R_+$ with $\int_0^\infty f^2=1$. The operator inequality \eqref{eq:CS-operator} implies that for any $x\in\R^d$ one has \eqref{eq:kernelbound}, understood as an operator inequality in $\mathcal H$. Denoting by $\lambda_n(T)$ the $n$-th eigenvalue, in decreasing order and taking multiplicities into account, of a non-negative operator $T$, we infer from \eqref{eq:kernelbound}, the variational principle and the computation \eqref{eq:rho-} that for any $n\in\N$,
$$
\lambda_n(\gamma(x,x)) \leq (1+\eps) \lambda_n((f(s/p^2)\gamma f(s/p^2))(x,x)) + (1+\eps^{-1})  s^{\frac{d}{2}}  \frac{|B_1|}{(2\pi)^d} A_f .
$$
At this stage we can optimize over $\eps>0$ and obtain
\begin{align} \label{eq:CS-kernelop}
\sqrt{\lambda_n(\gamma(x,x))} \le \sqrt{\lambda_n((f(s/p^2)\gamma f(s/p^2))(x,x))} + \sqrt{(1-f(s/p^2))^2 (x,x)}. 
\end{align}
Thus,
\begin{align} \label{eq:bound-Afop}
\lambda_n((f(s/p^2)\gamma f(s/p^2))(x,x))  \ge  \left[\sqrt{\lambda_n(\gamma(x,x))}-\sqrt{s^{\frac{d}{2}}  \frac{|B_1|}{(2\pi)^d} A_f} \, \right]_+^2.
\end{align}
For fixed $n$ (and $x$) we obtain after integration over $s$,
\begin{align*}
\int_0^\infty \lambda_n((f(s/p^2)\gamma f(s/p^2))(x,x))\,ds 
\geq \lambda_n(\gamma(x,x))^{1+\frac{2}{d}} \left( \frac{|B_1|}{(2\pi)^d} A_f\right)^{-\frac{2}{d}} \frac{d^2}{(d+2)(d+4)} .
\end{align*}
Summing over $n$ and integrating with respect to $x$ we obtain by \eqref{eq:kinetic-representationop}
\begin{align*}
\Tr(-\Delta \gamma) & \geq \int_{\R^d} \sum_n \int_0^\infty \lambda_n((f(s/p^2)\gamma f(s/p^2))(x,x))\,ds \\
& \geq \int_{\R^d} \tr \left( \gamma(x,x)^{1+\frac{2}{d}}\right) \d x \left( \frac{|B_1|}{(2\pi)^d} A_f\right)^{-\frac{2}{d}} \frac{d^2}{(d+2)(d+4)} .
\end{align*}
The proposition now follows in the same way as Proposition \ref{prop:gRumin}.
\end{proof}

\begin{remark}\label{rem:liftingfrac}
The same proof yields the operator-valued analogue of Theorem \ref{thm:fractional}. Since there seems to be no analogue of the following proposition for $(-\Delta)^\sigma$ with $\sigma\neq 1$, we do not write this out.
\end{remark}

In order to obtain good constants in higher dimensions we recall the following bound which is essentially due to Laptev and Weidl \cite{LapWei-00}. The extension to $d_1\geq 2$, which is not needed here, but is interesting in its own right, is due to \cite{Hundertmark-02}.

\begin{proposition}\label{laptevweidl}
For any integers $1\leq d_1<d$,
$$
L_{1,d}^{\rm op}/L_{1,d}^{\rm cl} \leq L_{1,d_1}^{\rm op} / L_{1,d_1}^{\rm cl} \,.
$$
\end{proposition}

In particular, taking $d_1=1$ and using the bound from Proposition \ref{prop:gRuminlift} together with \eqref{eq:L<=Lop} we obtain the following bound.

\begin{corollary}
For any $d\geq 1$, $L_{1,d}/L_{1,d}^{\rm cl} \le L_{1,d}^{\rm op}/L_{1,d}^{\rm cl}\leq 1.618435$.
\end{corollary}

The proof of Proposition \ref{laptevweidl} is by now standard, but we sketch it for the sake of completeness. We need the following more general family of Lieb--Thirring inequalities,
\bq \label{eq:LT-gen-eigenvalueop}
\Tr[-\Delta + V]_ -^\alpha \le  L_{\alpha,d}^{\rm op} \int_{\R^d} \tr \left( V(x)_-^{\alpha+\frac{d}{2}}\right) \d x \,,
\eq
as well as the semi-classical constant
$$
L_{\alpha,d}^{\rm cl} = \frac{1}{(2\pi)^2}\int_{\R^d} (\eta^2-1)_-^{\alpha+\frac{d}{2}}\, d\eta =  \frac{\Gamma(\alpha+1)}{(4\pi)^{d/2}\, \Gamma(\alpha+d/2+1)} \, \cdot
$$
where again  $V$ takes now values in the self-adjoint operators on some auxiliary separable Hilbert space $\mathcal{H}$ and its negative part $V(x)_-$ is in the $\alpha+\tfrac{d}{2}$ von Neumann--Schatten ideal, $\tr$ denotes the trace over  $\mathcal{H}$, and $\Tr$ the trace over $L^2(\R^d,\mathcal{H})= L^2(\R^d)\otimes \mathcal{H}$.  

The celebrated result by Laptev and Weidl \cite{LapWei-00} says that $L_{\alpha,d}^{\rm op} = L_{\alpha,d}^{\rm cl}$ for any $\alpha\geq 3/2$ and any $d \ge 1$. (For $d=1$, $\alpha=3/2$ and in the scalar case, this was shown in the original paper of Lieb and Thirring \cite{LieThi-76}.)

\begin{proof}[Proof of Proposition \ref{laptevweidl}]
We follow the argument in \cite{Hundertmark-02} closely: 
Let $d=d_1+d_2$ and decompose accordingly $x=(x_1,x_2)$ with $x_1\in\R^{d_1}$ and $x_2\in\R^{d_2}$ and $-\Delta=-\Delta_1-\Delta_2$. Let $V$ be a function on $\R^d$ taking values in the self-adjoint operators in some Hilbert space $\mathcal H$. For any $x_1\in\R^{d_1}$ we can consider $W(x_1)=-\Delta_{2} + V(x_1,\cdot)$ as a self-adjoint operator in $\tilde{\mathcal H} = L^2(\R^d;\mathcal H)$. Thus, by the operator-valued LT inequality on $\R^{d_1}$,
\begin{align*}
\Tr[-\Delta + V]_ - = \Tr_{L^2(\R^{d_1})} [-\Delta_1 + W]_-
\le L_{1,d_1}^{\rm op} \int_{\R^{d_1}} \Tr_{L^2(\R^{d_2};\mathcal H)} \left( W(x_1)_-^{1+\frac{d_1}{2}}\right) \d x_1 \,.
\end{align*}
Since $1+\frac{d_1}{2} \geq \frac{3}{2}$ the bound from \cite{LapWei-00} implies, for any $x_1\in\R^{d_1}$,
$$
\Tr_{L^2(\R^{d_2};\mathcal H)} \left( W(x_1)_-^{1+\frac{d_1}{2}}\right)
\leq L_{1+\frac{d_1}{2},d_2}^{\rm cl} \int_{\R^{d_2}} \tr \left( V(x_1,x_2)_-^{1+\frac{d}{2}} \right) \d x_2 \,.
$$
Combining the last two inequalities and observing that
$$
L_{1,d_1}^{\rm cl} L_{1+\frac{d_1}{2},d_2}^{\rm cl} = L_{1,d}^{\rm cl}
$$
(see \cite{Hundertmark-02} for a non-computational proof of this identity), we obtain the claimed inequality.
\end{proof}

%%%%%%%%%%%%%%%%%%%

\section{Low momentum averaging}

%To go beyond Rumin's method, we have to do something else for the triangle inequality.  Our idea is to average over the momentum before doing the triangle inequality, in order to gain some cancelation. For example, instead of using 
%\begin{align}
%\left|\int_0^E \sum_{j} u_i(x) u_i^{e-}(x) \d e \right| \le  \int_0^E \sqrt{\rho_\gamma(x)}\sqrt{\rho^{e-}(x)} \d e
%\end{align}
%as in the previous section, we can use
%\begin{align}
%\left|\int_0^E \sum_{j} u_i(x) u_i^{e-}(x) \d e \right| \le   \sqrt{\rho_\gamma(x)}\sqrt{\sum_j \left|\int_0^E u_j^{e-}(x) \d e \right|^2 }.
%\end{align}
%which is better by Minkowski's inequality. We can actually push forward this advantage by taking a weight function. This gives 
%

Our main idea to improve the estimate in Proposition \ref{prop:gRumin} is to average over low momenta $s\le E$ before using the Cauchy--Schwarz inequality  \eqref{eq:CS-operator}. We will actually push forward this idea by adding a weight function. This leads to

\begin{proposition} \label{prop:momentum} For $d\ge 1$,  the best constant in the Lieb--Thirring inequality \eqref{eq:LT} satisfies
\begin{align} \label{eq:momentumbound}
K_d/K_d^{\rm cl} \ge \frac{d 2^{4/d}}{(d+2)^{1+4/d} C_d^{2/d}} \,,
\end{align}
where
\begin{align} \label{eq:inf-Cfl}
\mathcal{C}_d &:= \inf \left\{ \left( \int_0^\infty \varphi^2 \right)^{d/2} \frac{d}{2}\int_0^\infty \frac{\Big(1-\int_0^\infty \varphi(s) f(st) \d s\Big)^2}{t^{1+\frac{d}{2}}}  \d t \right\}
\end{align}
with the infimum taken over all functions $f,\varphi:\R_+\to \R_+$ satisfying $\int_0^\infty f^2=\int_0^\infty \varphi =1$. 

In particular, when $d=1$ we have $K_1/K_1^{\rm cl} \ge 0.471851$ and $L_{1,1}/L_{1,1}^{\rm cl} \le 1.455786$.
\end{proposition}

\begin{proof} Let $f,\varphi:\R_+\to \R_+$ satisfy $\int_0^\infty f^2=\int_0^\infty \varphi=1$. Recall the momentum decomposition \eqref{eq:kinetic-representation}. We have for any $\psi\in L^2(\R^d)$, $s,s'\in(0,\infty)$,
$$
\langle\psi,f(s/p^2)\gamma f(s'/p^2)\psi\rangle \leq \sqrt{ \langle\psi,f(s/p^2)\gamma f(s/p^2)\psi\rangle }\ \sqrt{ \langle\psi,f(s'/p^2)\gamma f(s'/p^2)\psi\rangle } \,,
$$
and therefore, for every $E>0$,
\begin{align*}
& \int_0^\infty \int_0^\infty \varphi(s/E) \langle\psi,f(s/p^2)\gamma f(s'/p^2)\psi\rangle \varphi(s'/E) \,ds\,ds' \\
& \leq \left( \int_0^\infty \varphi(s/E)  \sqrt{ \langle\psi,f(s/p^2)\gamma f(s/p^2)\psi\rangle } \d s\right)^2 \\
& \leq \left( \int_0^\infty \varphi(s/E)^2 \d s\right) \left( \int_0^\infty \langle\psi,f(s/p^2)\gamma f(s/p^2)\psi\rangle \d s \right).
\end{align*}
This implies that we have the operator inequality 
\begin{align} \label{eq:oper-bound-lam-f}
& \left( \int_0^\infty \varphi^2(s) \d s \right) \left(  \int_0^\infty f(s/p^2)\gamma f(s/p^2) \d s  \right)\nn\\
&=E^{-1}\left( \int_0^\infty \varphi^2(s/E) \d s \right) \left(  \int_0^\infty f(s/p^2)\gamma f(s/p^2) \d s  \right)\nn\\
&\ge E^{-1}\left(\int_0^\infty \varphi(s/E) f(s/p^2) \d s \right) \gamma \left( \int_0^\infty \varphi(s/E) f(s/p^2)  \d s \right)\nn\\
&= E g( E/p^2)  \gamma  g(E/p^2) 
\end{align}
with 
\bq \label{eq:def-g}
g(t):=\int_0^\infty \varphi(s) f(st) \d s.
\eq
Next, by the Cauchy--Schwarz  estimate similarly to \eqref{eq:CS-operator} (thanks to $0\le \gamma\le 1$) we have 
\begin{align} \label{eq:oper-bound-lam-f-a2}
\gamma \le (1+\eps) g( E/p^2)  \gamma  g(E/p^2) + (1+\eps^{-1}) (1-g( E/p^2))^2.
\end{align}
for every $\eps>0$. Combining \eqref{eq:oper-bound-lam-f} and \eqref{eq:oper-bound-lam-f-a2} we get
\begin{align} \label{eq:oper-bound-lam-f-a3}
E\gamma \le (1+\eps)  \left( \int_0^\infty \varphi^2 \right)  \left(  \int_0^\infty f(s/p^2)\gamma f(s/p^2) \d s  \right)  + (1+\eps^{-1}) E(1-g( E/p^2))^2.
\end{align}
Transferring \eqref{eq:oper-bound-lam-f-a3} to a kernel bound, using  the same computation as in \eqref{eq:rho-}-\eqref{eq:Af},  and then  optimizing over $\eps>0$ we obtain
\begin{align} \label{eq:kernel-bound-lam-f}
\left( \int_0^\infty \varphi^2 \right) \int_0^\infty (f(s/p^2)\gamma f(s/p^2))(x,x) \d s \ge \left[\sqrt{E\gamma(x,x)}- \sqrt{E^{1+\frac{d}{2}}  \frac{|B_1|}{(2\pi)^d} A_g} \right]_+^2 \,.
\end{align} 
%It remains to maximize over $E>0$ on the right side of \eqref{eq:kernel-bound-lam-f}. Note that for every $a,b>0$, the function $E\mapsto \sqrt{Ea}-\sqrt{{E^{1+d/2}b}}$ attains it maximum at $\sqrt{a}=\left( 1+d/2\right) \sqrt{E^{d/2} b}$ with the maximal value 
%$$
%\left( \frac{2\sqrt{a}}{(d+2)\sqrt{b}}\right)^{2/d}  \frac{\sqrt{a} d}{d+2}. 
%$$
Then optimizing over $E>0$ leads to
\begin{align} \label{eq:key-decomposition-2}
\left( \int_0^\infty \varphi^2 \right) \int_0^\infty (f(s/p^2)\gamma f(s/p^2))(x,x) \d s &\ge \sup_{E>0} E \left[\sqrt{\gamma(x,x)}- \sqrt{E^{\frac{d}{2}}  \frac{|B_1|}{(2\pi)^d} A_g} \right]_+^2\nn\\
&= \gamma(x,x)^{1+2/d} \frac{(2\pi)^2}{|B_1|^{2/d}} \cdot \frac{2^{4/d}d^2}{(d+2)^{2+4/d} A_g^{2/d}}.
\end{align}
Inserting this into \eqref{eq:kinetic-representation} we conclude that 
\begin{align} \label{eq:key-decomposition-3}
 \Tr(-\Delta \gamma) \ge \left( \int_{\R^d} \gamma(x,x)^{1+2/d} \d x \right)  \frac{(2\pi)^2}{|B_1|^{2/d}} \cdot \frac{2^{4/d}d^2}{(d+2)^{2+4/d} A_g^{2/d}\left( \int_0^\infty \varphi^2 \right)},
\end{align}
namely the best constant in \eqref{eq:LT} satisfies
$$
K_d/K_d^{\rm cl} \le \frac{2^{4/d}d}{(d+2)^{1+4/d} A_g^{2/d}\left( \int_0^\infty \varphi^2 \right)} \cdot
$$
Optimizing over $f,\varphi$ leads to \eqref{eq:momentumbound}. 

When $d=1$, using the upper bound $\mathcal{C}_1\le 0.373556$ in Lemma \ref{lem:Cd} below, we obtain $K_1/K_1^{\rm cl}\ge 0.471851...$ and $L_{1,1}/L_{1,1}^{\rm cl}\le 1.455785...$. 
\end{proof}

We end this section with

\begin{lemma} \label{lem:Cd} When $d=1$, the constant $\mathcal{C}_d$ in \eqref{eq:inf-Cfl} satisfies
$$
\frac{1}{3} \le \mathcal{C}_1 \le 0.373556.
$$
\end{lemma}

%Although we could not compute the exact value of $\mathcal{C}_1$, our estimate is rather satisfactory. Of course there is still room for improvement. 

\begin{proof} Let $f,\varphi:\R_+\to \R_+$ satisfy $\int_0^\infty f^2=\int_0^\infty \varphi=1$. Denote $g$ as in \eqref{eq:def-g} and  $a:=\int_0^\infty \varphi^2$. By the Cauchy--Schwarz inequality
$$
g(t)=\int_0^\infty \varphi(s) f(st) \d s \le \left(\int_0^\infty \varphi^2(s) \d s \right)^{1/2} \left(\int_0^\infty f^2(ts) \d s \right)^{1/2} = \sqrt{\frac{a}{t}}.
$$
Therefore, when $d=1$ we get the desired  lower bound 
$$
a^{1/2} \int_0^\infty \frac{(1-g(t))^2}{2t^{3/2}} \d t \ge a^{1/2} \int_0^\infty \frac{\left[ 1-\sqrt{\frac{a}{t}}\right]_+^2}{2t^{3/2}} \d t = \frac{1}{3}.
$$

The upper bound on $\mathcal{C}_1$ requires an explicit choice of $(f,\varphi)$. The analysis from Section \ref{sec:momentum-decomposition}  suggests the following choice 
$$f(t)=(1+\mu t^{3/2})^{-1}, \quad \mu= \left[\frac{4\pi}{9\sqrt{3}}\right]^{3/2}, \quad \varphi(t)=5(1-t^{1/4})\1(t\le 1),$$
which gives $\mathcal{C}_{1} \le 0.381378$. We can do slightly better by taking
$$
f(t)= (1+\mu_0 t^{4.5})^{-0.25},  \quad \varphi(t)= c_0 \frac{(1-t^{0.36})^{2.1}}{1+t}\1(t\le 1)
$$
with $\mu_0$ and $c_0$ determined by $\int_0^\infty f^2= \int_0^\infty \varphi=1$, leading to  $\mathcal{C}_{1}\le 0.373556$. 
\end{proof}

%%%%%%%%%%%%%

\section{Lifting to higher dimensions. II}

In this section we proceed analogously to Section \ref{sec:lifting1} to extend  Proposition \ref{prop:momentum} to the operator-valued case.

\begin{proposition}\label{prop:momentumop} For $d\ge 1$, the best constant in the Lieb--Thirring inequality \eqref{eq:LTop} satisfies
\begin{align} \label{eq:momentumboundop}
K_d^{\rm op}/K_d^{\rm cl} \ge \frac{d 2^{4/d}}{(d+2)^{1+4/d} \mathcal{C}_d^{2/d}} 
\end{align}
with $\mathcal{C}_d$ from \eqref{eq:inf-Cfl}. In particular, when $d=1$ we have $K_1^{\rm op}/K_1^{\rm cl} \ge 0.471851$ and $L_{1,1}^{\rm op}/L_{1,1}^{\rm cl} \le 1.455786$.
\end{proposition}

Combining this proposition with Proposition \ref{laptevweidl} (for $d_1=1$) and \eqref{eq:L<=Lop} we obtain Theorem \ref{thm:main}. It remains to prove the proposition.

\begin{proof}  Let $f,\varphi:\R_+\to \R_+$ satisfy $\int_0^\infty f^2=\int_0^\infty \varphi=1$ and denote $g$ as in \eqref{eq:def-g}. We follow the proof of Proposition \ref{prop:momentum} to arrive at the operator inequality \eqref{eq:oper-bound-lam-f-a3}. As in the proof of Proposition \ref{prop:gRuminlift} this implies for any $x\in\R^d$ and $n\in\N$,
$$
E \lambda_n(\gamma(x,x)) \leq (1+\eps)  \left( \int_0^\infty \varphi^2 \right) \lambda_n \left(  \int_0^\infty (f(s/p^2)\gamma f(s/p^2))(x,x) \d s  \right)  + (1+\eps^{-1}) E^{1+\frac{d}{2}}  \frac{|B_1|}{(2\pi)^d} A_g.
$$
Optimizing over $\eps>0$ we obtain
\begin{align*}
\left( \int_0^\infty \varphi^2 \right) \lambda_n\left( \int_0^\infty (f(s/p^2)\gamma f(s/p^2))(x,x) \d s \right) \ge \left[\sqrt{E\lambda_n(\gamma(x,x))}- \sqrt{E^{1+\frac{d}{2}}  \frac{|B_1|}{(2\pi)^d} A_g} \right]_+^2 \,.
\end{align*}
Finally, optimizing over $E>0$ leads to
\begin{align*}
\left( \int_0^\infty \varphi^2 \right) \lambda_n\left(\int_0^\infty (f(s/p^2)\gamma f(s/p^2))(x,x) \d s \right) &\ge \sup_{E>0} E \left[\sqrt{\lambda_n(\gamma(x,x))}- \sqrt{E^{\frac{d}{2}}  \frac{|B_1|}{(2\pi)^d} A_g} \right]_+^2\nn\\
&= \lambda_n(\gamma(x,x))^{1+2/d} \frac{(2\pi)^2}{|B_1|^{2/d}} \cdot \frac{2^{4/d}d^2}{(d+2)^{2+4/d} A_g^{2/d}}.
\end{align*}
Inserting this into \eqref{eq:kinetic-representation} we conclude that
\begin{align*}
 \Tr(-\Delta \gamma) \ge \left( \int_{\R^d} \tr \left( \gamma(x,x)^{1+2/d} \right) \d x \right)  \frac{(2\pi)^2}{|B_1|^{2/d}} \cdot \frac{2^{4/d}d^2}{(d+2)^{2+4/d} A_g^{2/d}\left( \int_0^\infty \varphi^2 \right)}.
\end{align*}
Finally, it remains to optimize over $f,\varphi$ to obtain \eqref{eq:momentumboundop}. The numerical values when $d=1$ are obtained from the upper bound on $\mathcal{C}_1$ in Lemma \ref{lem:Cd}. 
\end{proof}

\section{Bounds with fractional operators}

The proof of Theorem \ref{thm:fractional} is essentially the same as that of Theorem \ref{thm:main} (except we do not use the lifting argument)  and we only sketch the major steps.

\begin{proof}[Proof of Theorem \ref{thm:fractional}]  Let $f:\R_+\to \R_+$ satisfy $\int_0^\infty f^2=1$. We have the analogue of \eqref{eq:kinetic-representation},
\begin{align}\label{eq:kinetic-representationfrac}
\Tr((-\Delta)^\sigma \gamma)= \int_{\R^d} \left[ \int_0^\infty (f(s/|p|^{2\sigma})\gamma f(s/|p|^{2\sigma}))(x,x) \d s \right] \d x.
\end{align}
Using the Cauchy--Schwarz inequality as in \eqref{eq:CS-operator} with a parameter $\eps>0$ and optimizing over this parameter we obtain a generalization of \eqref{eq:CS-kernel}, 
\begin{align} \label{eq:CS-kernelfrac}
\sqrt{\gamma(x,x)} \le \sqrt{(f(s/|p|^{2\sigma})\gamma f(s/|p|^{2\sigma}))(x,x)} + \sqrt{(1-f(s/|p|^{2\sigma}))^2 (x,x)}
\end{align}
for all $x\in\R^d$. We now compute
\begin{align} \label{eq:rho-frac}
(1-f(s/|p|^{2\sigma}))^2 (x,x) = s^{\frac{d}{2\sigma}}  \frac{|B_1|}{(2\pi)^d} A_f^{(\sigma)}
\end{align}
where 
\begin{align} \label{eq:Affrac}
A_f^{(\sigma)}:=  \frac{d}{2\sigma}\int_0^\infty \frac{(1-f(t))^2}{t^{1+\frac{d}{2\sigma}}}  \d t.
\end{align}
Consequently, we deduce from \eqref{eq:CS-kernelfrac} that
\begin{align} \label{eq:bound-Affrac}
(f(s/|p|^{2\sigma})\gamma f(s/|p|^{2\sigma}))(x,x)  \ge  \left[\sqrt{\gamma(x,x)}-\sqrt{s^{\frac{d}{2\sigma}}  \frac{|B_1|}{(2\pi)^d} A_f^{(\sigma)}} \, \right]_+^2.
\end{align}
Inserting \eqref{eq:bound-Affrac} into \eqref{eq:kinetic-representationfrac} and integrating over $s>0$ lead to
\begin{align} \label{eq:Rumin-boundfrac}
\Tr((-\Delta)^\sigma \gamma) \ge \left( \int_{\R^d}  \gamma(x,x)^{1+\frac{2\sigma}{d}} \d x \right) \left( \frac{|B_1|}{(2\pi)^d} A_f^{(\sigma)}\right)^{-\frac{2\sigma}{d}} \frac{d^2}{(d+2\sigma)(d+4\sigma)} .
\end{align}
Thus,
\begin{align} \label{eq:Rumin-bound-Kfrac}
K_{d,\sigma}/K_{d,\sigma}^{\rm cl} \ge \frac{d}{d+4\sigma}  \left( A_f^{(\sigma)}\right)^{-\frac{2\sigma}{d}}. 
\end{align}
Lemma \ref{lem:min-Af} provides the minimium value of $A_f^{(\sigma)}$ optimized over $f$ with $\int_0^\infty f^2 =1$. This leads to the first desired bound 
\begin{align} \label{eq:kinetic-Rumin-frac-1}
K_{d,\sigma}/K_{d,\sigma}^{\rm cl}\ge  \frac{d}{d+4\sigma} 
	\left[ \frac{(d+2\sigma)^2 \sin\left( \frac{2\pi\sigma}{d+2\sigma} \right)}{2\pi\sigma d} \right]^{1+\frac{2\sigma}{d}}\, .
\end{align}

Next, we introduce $\varphi:\R_+\to \R_+$ satisfy $\int_0^\infty \varphi=1$ and denote $g$ as in \eqref{eq:def-g}. Then proceeding as in \eqref{eq:oper-bound-lam-f-a3} we have the operator inequality
\begin{align*} %\label{eq:oper-bound-lam-f-a3-frac}
E\gamma \le (1+\eps)  \left( \int_0^\infty \varphi^2 \right)  \left(  \int_0^\infty f(s/|p|^{2\sigma})\gamma f(s/|p|^{2\sigma}) \d s  \right)  + (1+\eps^{-1}) E(1-g( E/|p|^{2\sigma}))^2.
\end{align*}
Transfering the latter  to a kernel bound, using the same computation as in \eqref{eq:rho-frac}-\eqref{eq:Affrac},  and  optimizing over $\eps>0$ and then $E>0$ we obtain the following analogue of \eqref{eq:key-decomposition-2},
\begin{align} \label{eq:key-decomposition-2-frac}
&\left( \int_0^\infty \varphi^2 \right) \int_0^\infty (f(s/|p|^{2\sigma})\gamma f(s/|p|^{2\sigma}))(x,x) \d s \nn\\
&\ge \sup_{E>0} E \left[\sqrt{\gamma(x,x)}- \sqrt{E^{\frac{d}{2\sigma}}  \frac{|B_1|}{(2\pi)^d} A_g^{(\sigma)}} \right]_+^2\nn\\
&= \gamma(x,x)^{1+\frac{2\sigma}{d}} \left( \frac{|B_1|}{(2\pi)^d} A_g^{(\sigma)}\right)^{-\frac{2\sigma}{d}}  \left( \frac{d}{d+2\sigma}\right)^2  \left( \frac{2\sigma}{d+2\sigma}\right)^{\frac{4\sigma}{d}}. 
 \end{align}
Inserting \eqref{eq:key-decomposition-2-frac} into \eqref{eq:kinetic-representationfrac}, and then optimizing over $f,\varphi$ we arrive at 
$$
K_{d,\sigma}/K_{d,\sigma}^{\rm cl} \ge  \frac{d}{d+2\sigma}  \left( \frac{2\sigma}{d+2\sigma}\right)^{\frac{4\sigma}{d}} \left( A_g^{(\sigma)}\right)^{-\frac{2\sigma}{d}} \left( \int_0^\infty \varphi^2 \right)^{-1}
$$
Optimizing over $f,\varphi$ gives the second desired estimate 
\begin{align} \label{eq:kinetic-Rumin-frac-2}
K_{d,\sigma}/K_{d,\sigma}^{\rm cl}\ge  \frac{d}{d+2\sigma}  \left( \frac{2\sigma}{d+2\sigma}\right)^{\frac{4\sigma}{d}} \mathcal{C}_{d,\sigma}^{-\frac{2\sigma}{d}}
\end{align}
with $\mathcal{C}_{d,\sigma}$ given in \eqref{eq:inf-Cfl-frac}. 

Finally, in the physical case $\sigma=1/2$ and $d=3$, by taking the trial choice 
$$
f(t)=(1+\mu_0 t^{10})^{1/4}, \quad \varphi(t)=c_0 (1-t^2)^4 \1(t\le 1)
$$ 
with $\mu_0$ and $c_0$ determined by $\int_0^\infty f^2= \int_0^\infty \varphi=1$, we obtain $\mathcal{C}_{d,\sigma} \le 0.046736$, which implies $K_{d,\sigma}/K_{d,\sigma}^{\rm cl}\ge 0.826297$ by \eqref{eq:kinetic-Rumin-frac-2}.
\end{proof}


\begin{thebibliography}{18}

\bibitem{AbrSte-64} M. Abramowitz, and I. A. Stegun, Handbook of Mathematical Functions with Formulas, Graphs, and Mathematical Tables. Applied Mathematics Series 55 (1964). United States Department of Commerce, National Bureau of Standards. 

\bibitem{BN}
  Bak J., Newman D.J.:
  \textit{Complex Analysis}, Springer, 2010.


\bibitem{BlaStu-96} Ph. Blanchard and J. Stubbe, Bound states for Schr\"odinger Hamiltonians: Phase
Space Methods and Applications. Rev. Math. Phys., 35, 504-547 (1996)

\bibitem{CarFraLie-14} E. A. Carlen, R. L. Frank and E. H. Lieb, Stability estimates for the lowest eigenvalue of a Schrödinger operator, Geom. Funct. Anal. 24 (2014), no. 1, 63-84.

\bibitem{Daubechies-83} I. Daubechies, An uncertainty principle for fermions with generalized kinetic
energy, Commun. Math. Phys. 90, 511-520 (1983)

\bibitem{DolLapLos-08} J. Dolbeault, A. Laptev and M. Loss, Lieb-Thirring inequalities with improved constants,
J. Eur. Math. Soc. 10 (2008),  1121-1126.

\bibitem{DysLen-67} F. J. Dyson and A. Lenard, Stability of matter. I, J. Math. Phys. 8 (1967),  423-434; II. J. Math. Phys. 9 (1968),  698-711.


\bibitem{EdeFoi-91} A. Eden and C. Foias, A simple proof of the generalized Lieb-Thirring
inequalities in one-space dimension, J. Math. Anal. Appl. 162 (1991), 250-254.


\bibitem{Frank-14} R. L. Frank, Cwikel's theorem and the CLR inequality. J. Spectral Theory 4 (2014), no. 1, 1-21.

\bibitem{FraLewLieSei-13} R. L. Frank, M. Lewin, E.H. Lieb and R. Seiringer, A positive density analogue of the Lieb-Thirring inequality. Duke Math. J. 162 (2013), no. 3, 435-495.
%
%\bibitem{FraSei-12} R. L. Frank and R. Seiringer, Lieb-Thirring inequality for a model of particles
%with point interactions, J. Math. Phys., 53 (2012),  095201.
%
%\bibitem{Kroger-94} P. Kroger, Estimates for sums of eigenvalues of the Laplacian, J. Funct. Anal. 126 (1994), 217-227.
%
%

\bibitem{HO-77} M. Hoffmann-Ostenhof and T. Hoffmann-Ostenhof, Schr\"odinger inequalities
and asymptotic behavior of the electron density of atoms and molecules, Phys. Rev.
A 16 (1977),  1782-1785.



\bibitem{Hundertmark-02} D. Hundertmark, On the number of bound states for Schr\"o�dinger operators with operator-valued potentials, Ark. Mat. 40 (2002), 73-87.

\bibitem{HunLapWei-00} D. Hundertmark, A. Laptev and T. Weidl, New bounds on the Lieb-Thirring constants, Invent. Math. 140 (2000), pp. 693-704.


\bibitem{Keller-61} J. B. Keller, Lower bounds and isoperimetric inequalities for eigenvalues of the Schr\"odinger equation. J. Mathematical Phys. 2 (1961), 262-266.

\bibitem{Laptev-97} A. Laptev, Dirichlet and Neumann Eigenvalue Problems on Domains in Euclidean Spaces, J. Func. Anal. 151 (1997), 531-545.

\bibitem{LapWei-00} A. Laptev and T. Weidl, Sharp Lieb-Thirring inequalities in high dimensions. Acta Math., 184 (2000), 87-111.


%\bibitem{LewLie-15} M. Lewin and E. H. Lieb, Improved Lieb-Oxford exchange-correlation inequality with gradient correction. Phys. Rev. A 91 (2015), 022507. 
%
%\bibitem{Levitt-14} A. Levitt, Best constants in Lieb-Thirring inequalities: a numerical investigation. Journal of Spectral Theory 4.1 (2014),  153-175
%
%\bibitem{LiYau-83} P. Li and S.T. Yau, On the Schr\"odinger equation and the eigenvalue problem. Commun. Math.
%Phys. 88 (1983),  309-318.
%
%\bibitem{Lieb-81} E. H. Lieb, Thomas-Fermi and related theories of atoms and molecules, Rev. Mod. Phys. 53 (1981),  603-641.
%
\bibitem{Lieb-84} E.H. Lieb, On characteristic exponents in turbulence. Commun. Math. Phys. 82 (1984), pp. 473-480.




\bibitem{LieLos-01} E.~H. Lieb and M.~Loss, Analysis, Second edition, Graduate Studies in Mathematics, American Mathematical Society, Providence, RI, 2001.
	
\bibitem{LieThi-75} E. H. Lieb and W. E. Thirring, Bound on kinetic energy of fermions which proves
stability of matter, Phys. Rev. Lett. 35 (1975),  687-689.

\bibitem{LieThi-76} E. H. Lieb and W. E. Thirring, Inequalities for the moments of the eigenvalues of the Schr\"odinger Hamiltonian and their relation to Sobolev inequalities, in Studies in Mathematical Physics,
Princeton University Press, 1976,  269-303.

\bibitem{LunSol-13} D. Lundholm and J. P. Solovej, Hardy and Lieb-Thirring inequalities for anyons. Commun. Math. Phys. 322 (2013),  883-908. 


\bibitem{Nam-18} P. T. Nam, Lieb-Thirring inequality with semiclassical constant and gradient error term. J. Funct. Anal. 274 (2018), 1739-1746. 





%\bibitem{LunNamPor-16} D. Lundholm,  P. T. Nam, and F. Portmann, Fractional Hardy-Lieb-Thirring and related inequalities for interacting systems. Arch. Rational Mech. Anal. 219 (2016), 1343-1382.
%
%\bibitem{LunPorSol-15} D. Lundholm, F. Portmann, and J. P. Solovej, Lieb-Thirring bounds for interacting
%Bose gases, Commun. Math. Phys. 335 (2015),  1019-1056.
%
%
%\bibitem{LunSol-13} D. Lundholm and J. P. Solovej, Hardy and Lieb-Thirring inequalities for anyons,
%Commun. Math. Phys. 322 (2013),  883-908.
%
%\bibitem{LunSol-14} D. Lundholm and J. P. Solovej, Local exclusion and Lieb-Thirring inequalities for intermediate and fractional statistics, Ann. Henri Poincar\'e 15 (2014), 1061-1107.
%
%

\bibitem{Rumin-11} A. Rumin, Balanced distribution-energy inequalities and related entropy bounds, Duke
Math. J., 160 (2011), 567-597.

\end{thebibliography}
\end{document}